\documentclass[letterpaper, 10 pt, twocolumn]{ieeeconf}

\IEEEoverridecommandlockouts

\usepackage{mathtools,amsthm,amssymb,amsmath,amsfonts,bbm,bm}
\usepackage{xcolor}
\usepackage[mathscr]{euscript}
\usepackage{graphicx,psfrag}
\usepackage{float}
\usepackage{makecell}

\usepackage{multirow}
\usepackage{tcolorbox}

\usepackage{subfigure}
\usepackage[english]{babel}

\usepackage{lipsum} 
\usepackage{cite}

\usepackage[hidelinks]{hyperref}
\usepackage{dsfont}

\let\labelindent\relax
\usepackage[shortlabels]{enumitem}

\usepackage{epsfig}

\usepackage[linesnumbered,ruled,vlined]{algorithm2e}
\SetKwInput{KwInput}{Input}                
\SetKwInput{KwOutput}{Output}              

\theoremstyle{plain}
\newtheorem{theorem}{Theorem}
\newtheorem{lemma}{Lemma}

\newtheorem{proposition}{Proposition}

\theoremstyle{definition}

\newtheorem{assumption}{Assumption}

\newtheorem{problem}{Problem}

\theoremstyle{remark}
\newtheorem{remark}{Remark}

\DeclareMathOperator{\im}{im}

\newcommand{\norm}[1]{\ensuremath{\left\| #1 \right\|}}

\newcommand{\calA}{\ensuremath{\mathcal{A}}}

\newcommand{\calC}{\ensuremath{\mathcal{C}}}
\newcommand{\calD}{\ensuremath{\mathcal{D}}}

\newcommand{\calO}{\ensuremath{\mathcal{O}}}
\newcommand{\calQ}{\ensuremath{\mathcal{Q}}}
\newcommand{\calS}{\ensuremath{\mathcal{S}}}

\newcommand{\bA}{\ensuremath{\boldsymbol{A}}}
\newcommand{\bB}{\ensuremath{\boldsymbol{B}}}
\newcommand{\bC}{\ensuremath{\boldsymbol{C}}}
\newcommand{\bQ}{\ensuremath{\boldsymbol{Q}}}
\newcommand{\bze}{\ensuremath{\boldsymbol{\zeta}}}

\newcommand{\ups}{\ensuremath{\boldsymbol{\Upsilon}}}

\newcommand{\s}{\ensuremath{\star}}
\newcommand{\DA}{\ensuremath{\Delta A}}
\newcommand{\DB}{\ensuremath{\Delta B}}

\newcommand{\R}{\ensuremath{\mathbb R}}

\newcommand{\EP}{\hspace*{\fill} $\square$\smallskip\noindent}
\newcommand{\bse}{\begin{subequations}}
	\newcommand{\ese}{\end{subequations}}
\def\be{\begin{equation}}
\def\ee{\end{equation}}

\newcommand{\bbm}{\begin{bmatrix}}
	\newcommand{\ebm}{\end{bmatrix}}

\usepackage[prependcaption,colorinlistoftodos]{todonotes}

\usepackage[normalem]{ulem}

\title{\LARGE \bf
	Preserving Privacy in Cloud-based Data-Driven Stabilization
}

\author{Teimour Hosseinalizadeh, Nima Monshizadeh
	\thanks{The authors are with the Engineering and Technology
		Institute, University of Groningen, Groningen, The Netherlands. Email: {\tt\{t.hosseinalizadeh,  n.monshizadeh\}@rug.nl}. \\
        This work has been submitted to the IEEE for possible publication. Copyright may be transferred without notice, after which this version may no longer be accessible.
}}

\begin{document}
	
	\maketitle
	\thispagestyle{empty}
	\pagestyle{empty}
	
	\begin{abstract}
		In the recent years, we have observed three significant trends in control systems:  a renewed interest in  data-driven control design, the abundance of cloud computational services and the importance of preserving privacy for the system under control. Motivated by these factors, this work investigates privacy-preserving outsourcing for the design of a stabilizing controller for unknown linear time-invariant systems. 
		The main objective of this research is  to  preserve the privacy for the system dynamics by designing an outsourcing mechanism.
		To achieve this goal, we propose a scheme that combines transformation-based techniques and robust data-driven control design methods.
		The  scheme preserves the privacy of both the open-loop and closed-loop system matrices while stabilizing the system under control.
		The scheme is applicable to both data with and without disturbance and is lightweight in terms of computational overhead.
		Numerical investigations for a  case study  demonstrate the impacts of our mechanism and its role in hindering malicious adversaries from achieving their goals.
	\end{abstract}

	\section{Introduction}\label{sec:intro}
	Cloud-based services  enable access to computing resources such as networks, servers, storage and  applications that can be rapidly provisioned and released with minimal management effort or service provider interactions \cite{mell2011nist}. 
	Researchers in control systems have documented the enhancement of conventional control technologies by leveraging cloud-based services across numerous industries.  
	For instance, cloud-based control offers a $53\%$ energy saving and a $36\%$ improvement in thermal comfort for an office building \cite{Buildingdrgovna2020cloud}, reduces electrical energy consumption in geothermal fields \cite{GeoThermalStoffel2022cloud} and lowers peak demand during the cooling season in public schools \cite{K12SchoolWoo2023first}. 
	Nonetheless, numerous security issue arises in cloud-based control due to the nature of cloud computing, one of which is the unauthorized access of could to privacy-sensitive  parameters of the system and controller \cite{xia2022brief}.
	
	\par \textit{Related studies:} 
	Since the advent of privacy concerns in control systems \cite{le2013differentially}, preserving privacy has become the focus of many studies using a wide range of methods and for a wide range of applications. For instance, \cite{xu2024cloud} proposes a cloud-based load frequency control in power systems using homomorphic encryption, while \cite{pan2022privacy} employs encryption to preserve privacy in platooning control of vehicular systems. Additionally, \cite{chen2024privacy} addresses the problem of distributed economic dispatch in microgrids using edge-based additive perturbations, and \cite{qi2024full} explores outsourcing controllers for switched LPV systems using differential privacy. For a  review of privacy-preserving methods in control systems, we refer to the surveys in \cite{schluter2023brief} and \cite{han2018privacy}
	
	\par In the literature on privacy for control systems there have been numerous studies with focus on preserving the privacy for the \textit{system model or parameters} against adversaries. For instance, 
	\cite{le2013privacy} considers the problem of privacy-preserving releasing of a dynamic model describing the aggregate input-output dynamics of subsystems,
	\cite{bottegal2017preserving} motivated by the model of the system as a trade secret  designs optimal filters for constructing additive noise, \cite{ziemann2020parameter} considers preserving privacy for the state matrix in a linear quadratic problem, \cite{prakash2024privacy} studies the problem of securely outsourcing the solution of algebraic Riccati equations  
	while \cite{katewa2019differential} considers designing synthetic noise for a multiagent systems monitored by a control center and in the presence of an intruder. 
	The main difference between these studies and the current work is that in these papers, the system model is known when designing a privacy-preserving mechanism.  
	Furthermore, with the exception of \cite {prakash2024privacy}, these studies utilize noise which its amount is adjusted based on differential privacy or Fisher information matrix \cite{ziemann2020parameter} to guarantee privacy against adversaries. This introduces the  well-known privacy-performance trade-off.

	\par With the ever-increasing complexities of  dynamic systems, the controller or privacy-preserving mechanism designer  may not have the accurate system model to accomplish their tasks.  This has been a key motivation (among others) for the recent renewed interest in data-driven control, which uses input and output trajectories of a system rooted in behavioral framework \cite{willems2005note}  (see \cite{de2019formulas} and \cite{coulson2019data}).
	The literature on secure control has also adopted this framework for which we refer to
	\cite{russo2021poisoning,alisic2023model,anand2023data}.
	\par This brings us to the second class of studies with the focus on preserving privacy for the system model, which this research also belongs to,  and that is the case where the designer of the privacy-persevering mechanism \textit{does not know}  the model of the system.  For instance, 
	\cite{suh2021sarsa}  considers a cloud-based reinforcement learning (RL) where updating the value function is outsourced to the cloud, 
	\cite{alexandru2020towards} studies cloud-based data-driven  model predictive control  inspired by the behavioral framework,  
	and \cite{mukherjee2021secure}  presents a secure RL-based control method for unknown linear time-invariant (LTI) systems where a dynamic camouflaging technique is used to misguide the eavesdropper.  
	The studies \cite{suh2021sarsa} and \cite{alexandru2020towards}  rely on CKKS as a Fully Homomorphic Encryption (FHE) to prevent the cloud from inferring  sensitive information. While FHE provides strong privacy guarantees it can introduce additional computational overhead and possibly delays as they report in \cite{suh2021sarsa}.     
	The paper \cite{mukherjee2021secure} draws on dynamical systems theory  to hide sensitive information, however it only considers an eavesdropper in the learning phase of the algorithm and permits the cloud to  access  privacy-sensitive data in the design phase. This approach differs significantly from our problem structure, as well as from the studies in \cite{suh2021sarsa} and \cite{alexandru2020towards}, where the cloud is not trusted 

	\par This research aims to design a data-driven stabilizing  controller for an unknown LTI system by following the paradigm presented in \cite{de2019formulas}.  
	The additional constraint we impose is the preservation of privacy for both open-loop and closed-loop system matrices when the controller design is performed by the cloud.		To achieve this goal, we propose a scheme that combines transformation-based methods and robust controllers design approaches. The scheme provides a privacy budget for the closed-loop system allowing the system designer to conceal the closed-loop system, even in the presence of unknown disturbances in the system's dynamics.  
	Furthermore, through extensive numerical simulations we demonstrate the strong effects of the proposed scheme in preventing active adversaries from performing bias injection attacks on the system. 
	To summarize, the contributions of this research are as follows:
	\begin{enumerate}
		\item Preserving privacy for  both open-loop and closed-loop system matrices while designing a stabilizing controller
		\item Providing a privacy budget for the closed-loop system  
		\item Guaranteeing  privacy for systems in the presence of unknown disturbances 
	\end{enumerate}
	Additionally, the scheme does not degrade the controller's performance,  unlike noise-based methods, and  is lightweight in terms of computational complexity.  
	
	\par \textit{Notation:}
	We denote the identity matrix of size $n$ by $I_n$, the zero matrix of size $n \times m$ by $0_{n \times m}$ and drop the subscript whenever the dimension is clear from the context. For a matrix $A \in \R^{n \times m}$,  we denote its induced 2-norm, and Moore-Penrose inverse by  $\norm{A}$, and $A^\dagger$, respectively. 
	By $A \succ0$ ($\succeq 0$), we mean $A$ is a positive (semi-) definite matrix and by $A^{1/2}$ we denote its unique positive (semi-) definite square root. We also denote $\begin{bsmallmatrix} A & B^\top \\ B & C\end{bsmallmatrix}$ by $\begin{bsmallmatrix} A & \s \\ B & C\end{bsmallmatrix}$.
	\par \textit{Organization:} The rest of this paper is organized as follows: Section \ref{sec:pre-prob-setup} presents preliminaries and problem formulation, Section \ref{sec:privacy-clean-sys} provides a privacy-preserving scheme for systems without disturbance while Section \ref{sec:privacy-disturb-sys} extends the scheme to systems with disturbance. Numerical simulations and concluding remarks are provided in Sections \ref{sec:case_study} and \ref{sec:conclusion}, respectively.

	\section{Preliminaries and Problem setup}\label{sec:pre-prob-setup}
	In this section we present the preliminaries and problem setup.
	\subsection{Preliminaries} 
	We frequently encounter a set $\Sigma$ of systems 
	parameterized through a quadratic matrix inequality (QMI) 
	as
	\be\label{eq:ellips-lemma}
	\begin{aligned}
		\Sigma :=\{\bbm {A} & {B}\ebm^\top :={Z}\mid {\bC}+{\bB}^{\top}{Z}+{Z}^{\top}{\bB}+ {Z}^{\top}{\bA}{Z} \preceq 0 \},
	\end{aligned}
	\ee
	where $\bA \succ 0 $ and ${\bB}^{\top}{\bA}^{-1}{\bB}-{\bC} \succeq 0$.
	The matrices $\bA$, $\bB$ and $\bC$ are generally known and dependent on the data collected from a system. 
	
	\par We are particularly interested in a controller $K$ that render all matrices $A+BK$, with $[A \; B]^\top\in \Sigma$,  Schur stable. 
	A necessary and sufficient conditions for the existence and design (in case of existence) of such a stabilizing controller is provided by the following lemma.

	\begin{lemma}\label{lem:stability-QMI}
		Let $\Sigma$ determine the set of linear systems given by the pair $(A, B)$. 
		Then there exist the matrices $K$ and $P\succ0$ such that
		$$
		(A + BK)P(A + BK)^\top - P \prec 0\quad  \text{for all} \bbm A & B\ebm^\top \in \Sigma
		$$
		if and only if there exist $Y$ and $P\succ 0$ such that 
		\be\label{eq:final_LMI-for_QMI}
		\begin{aligned}
			\bbm 
			-P -\bC & \s & \s \\
			0 & -P & \s \\
			{\bB} & \bbm
			P \\ Y\ebm & -{\bA}
			\ebm \prec 0.
		\end{aligned}
		\ee
		If \eqref{eq:final_LMI-for_QMI}  is feasible, then $K=YP^{-1}$.
	\end{lemma}
	\begin{proof}
		The proof is  given in  \cite[Thm.1]{bisoffi2022data}. 
	\end{proof}
	
	{In this manuscript, we ask for solving LMIs in the form of \eqref{eq:final_LMI-for_QMI}  by specifying the matrices $\bA$, $\bB$ and $\bC$ in the set $\Sigma$.}
	\subsection{Problem setup} \label{subsec:problem-setup}
	\begin{figure*}[t]
		\begin{center}
			\includegraphics[width=0.75\textwidth]{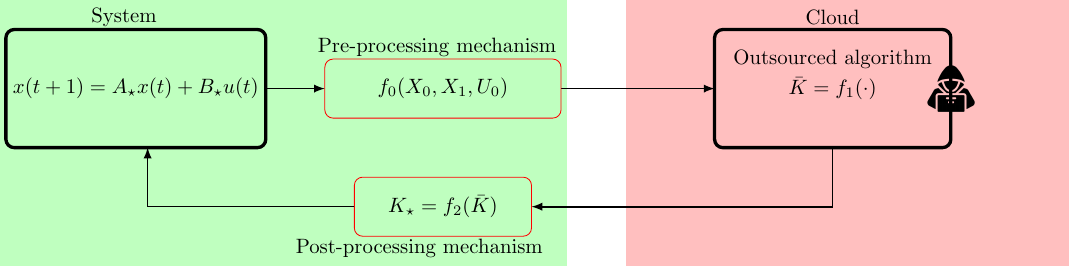}
			\caption{A framework for preserving privacy  in Cloud-based direct data-driven control. The pre-processing mechanism $f_0(\cdot)$ is applied to the data collected from system then the results are transmitted to Cloud where it executes a predefined algorithm $f_1(\cdot)$ for obtaining the controller $\bar{K}$. The post-processing mechanism $f_2(\cdot)$ modifies the obtained controller to the final controller $K_{\s}$ which is applied to the system.}\label{fig:setup}
		\end{center}
	\end{figure*}
	We consider a linear system described by equations of the form
	\be\label{eq:lin_sys}
	x(t+1) = A_{\s}x(t) + B_{\s}u(t),
	\ee
	with state $x \in \R^{n}$ and input $u \in \R^{m}$. In this study, the matrices $A_\s$ and $B_\s$ are real, constant and \textit{unknown}. Furthermore, we assume that $B_\s$ has full column rank. 
	\par We collect data form the system \eqref{eq:lin_sys} by applying the input sequence $u(0), u(1), \ldots, u(T-1)$ and measuring the states $x(0), x(1), \ldots, x(T)$. The collected data is organized into the matrices $(X_0,X_1, U_0)$ as follows:
	\be\label{eq:data-original-ol}
	\begin{aligned}
		& X_0 \coloneqq \bbm x(0) & x(1) & \cdots & x(T-1) \ebm \in \R^{n \times T},\\
		& X_1 \coloneqq \bbm x(1) & x(2) & \cdots & x(T) \ebm \in \R^{n \times T}, \\ 
		& U_0 \coloneqq \bbm u(0) & u(1) & \cdots & u(T-1) \ebm \in \R^{m \times T}.
	\end{aligned}
	\ee
	Note that that these data matrices satisfy 
	\be\label{eq:open-loop-data-system}
	X_1=A_\s X_0 + B_\s U_0.
	\ee
	We aim to have a  controller $u(t) = K_\s x(t)$ with $K_\s \in \R^{m \times n}$ to stabilize  system \eqref{eq:lin_sys} by using data \eqref{eq:data-original-ol}. To achieve this, we outsource the computation of this controller to a cloud computing service which we refer to as Cloud.
	
	\par Additionally, we want to preserve the privacy of the open-loop $(A_{\s}, B_{\s})$ and closed-loop matrix $A_\s +B_\s K_\s$   against Cloud.
	The main reason for this privacy constraint is that 
	these matrices typically contain sensitive information about the process under control. Such information can be leaked to curious adversaries as a result of outsourcing the control problem to Cloud. Moreover, knowledge of these matrices can be leveraged to design additional and potentially more powerful cyber-attacks. We refer the reader to Figure 1 in \cite{teixeira2015secure} for the role that model knowledge plays in designing zero dynamics, covert and bias injection attacks. Later, in Section \ref{sec:case_study}, we demonstrate that  how a successful bias injection attack relies on the accuracy of model knowledge. 
	
	\par Our target applications are  dynamical systems with low computational capabilities or in  extreme cases without any computational resources, meaning  we only have access to the system's sensors and actuators. 
	For complex systems (high dimensions in states and inputs) designing a data-driven controller can be computationally intensive.
	This challenge is further compounded in unstable systems, where large value measurements must be stored and processed.
	Additionally, in switched dynamical systems,  frequent data  measurements and controller updates may be required which introduces  significant computational overhead (see  \cite{rotulo2022online} for data-driven control of switched systems).
	
	\par We have the following assumption with regard to the collected data in this manuscript:
	\begin{assumption}\label{assum:rank-data}
		The matrix $\bbm X_0 \\ U_0\ebm$ in \eqref{eq:data-original-ol} has full row rank.
	\end{assumption}
	The assumption can be checked directly from the data, and it is 
	satisfied under persistently exciting inputs of sufficiently high order \cite{willems2005note}.
	\begin{remark}
		Assumption \ref{assum:rank-data} is aligned with the general assumption in the privacy literature that without any privacy-preserving mechanism an adversary should be considered capable of finding privacy-sensitive information (see, e.g., maxim number one in \cite[p.371]{van2014encyclopedia}). In particular, without any privacy-preserving mechanism in place, Cloud can identify the system matrices $(A_\s, B_\s)$ from data in \eqref{eq:data-original-ol} under Assumption \ref{assum:rank-data}. 
	\end{remark}
	\par Figure \ref{fig:setup} shows the sketch for the privacy-preserving scheme in this study. The functions $f_0$ (pre-processing mechanism) and $f_2$ (post-processing mechanism) and the algorithm $f_1$ (outsourced mechanism) are our design tools. 
	Data collected from the system undergoes pre-processing via $f_0(\cdot)$
	and is then transmitted to  Cloud, where the predefined algorithm $f_1(\cdot)$ computes the controller $\bar{K}$. The post-processing mechanism $f_2(\cdot)$ then adjusts $\bar{K}$ to the final controller $K_\s$
	which is applied to the system.
	We further determine the general properties of these tools in the following assumption:
	\begin{assumption}\label{assum:problem_structure} The followings are assumed:
		\begin{enumerate}[label=A2.\arabic*.]
			\item We design  $f_0$, $f_1$, and $f_2$ and Cloud executes $f_1$ honestly. \label{assum:A1}
			\item The function $f_0$ and $f_2$ are publicly known except for their secret parameters that are randomly chosen. \label{assum:A2}
		\end{enumerate}
	\end{assumption}
	Assumption \ref{assum:A1} determines  Cloud model in our setup. We further point out that while Cloud in this study is a passive adversary, it can share its obtained information with active adversaries or it can be compromised by these adversaries. 
	Assumption \ref{assum:A2} is aligned with Kirchhoff's principle which requires that an encryption scheme must be designed to be secure even if the adversary knows all the details of the scheme except for the key (see, e.g. \cite[p.7]{katz}).
	\par It should be noted that the functions $f_0$ and $f_2$ should impose less computational burden compared to $f_1$, as the former ones will be locally executed. 
	
	\par In this work, we say that  privacy for the matrices $A_\s$, $B_\s$ and $A_\s+B_\s K_\s$ is preserved if Cloud cannot uniquely identify their values based on the available information. Next, we formally present the problem of interest:
	\begin{problem}\label{problem:data-driven}
		Let the data \eqref{eq:data-original-ol} be available from the system \eqref{eq:lin_sys} and suppose that Assumptions \ref{assum:rank-data} and \ref{assum:problem_structure} hold. Design pre- and post-processing mechanisms  $f_0$ and $f_2$ and the algorithm $f_1$ such that:
		\begin{enumerate}[label=P\arabic*.]
			\item The closed-loop system $A_\s + B_\s K_\s$ is stable. \label{prob:stability}
			\item The open-loop system matrices $A_{\s}$ and $B_\s$ remain private against Cloud. \label{prob:privacy_open_loop}
			
			\item The closed-loop system $A_\s + B_\s K_\s$ remains private against Cloud. \label{prob:privacy_closed_loop}
		\end{enumerate}
	\end{problem}
	Next, we present our privacy-preserving scheme for the  system \eqref{eq:lin_sys} solving Problem \ref{problem:data-driven}. Subsequently, in Section \ref{sec:privacy-disturb-sys}, we discuss how the proposed scheme can be modified to cope with the scenario where the input-state data is corrupted by disturbances. 
	
	\section{{Preserving privacy for systems without disturbance}}\label{sec:privacy-clean-sys}
	In this section, we first provide a stabilizing controller while preserving privacy for the open-loop system.  Subsequently, we preserve privacy for the closed-loop system.
	\subsection{Preserving privacy for the open-loop system}
	Recall that the data set \eqref{eq:data-original-ol} has been collected from system \eqref{eq:lin_sys}. We apply a transformation given by
	\be\label{eq:pre-proc-func}
	V_0 := \bbm -(I +G_1)^{-1}F_1 & (I +G_1)^{-1}\ebm \bbm X_0 \\ U_0\ebm,
	\ee
	where $F_1 \in \R^{m \times n}$ and $G_1 \in \R^{m \times m}$ are randomly selected matrices such that $I+G_1$ is invertible. Noting \eqref{eq:open-loop-data-system}, applying the transformation \eqref{eq:pre-proc-func} yields
	\be\label{eq:sys-clean-trans}
	X_1 = {(A_\s + B_\s F_1)}X_0 + {(B_\s + B_\s G_1)}V_0.
	\ee
	While the transformation \eqref{eq:pre-proc-func} is performed on the collected data, which can be considered as the cyber part of a cyber-physical system, we explain in the remark below how it can also be regarded as a primary feedback for the system under control. 
	\begin{remark}
		The  transformation \eqref{eq:pre-proc-func} can be considered equivalently as first applying a primary feedback 
		$$
		u(t)= F_1x(t) + (I_m + G_1) v(t),
		$$
		to the system \eqref{eq:lin_sys} and then collecting data from the system 
		$$
		x(t+1) = (A_{\s} + B_\s F_1) x(t) + (B_{\s} + B_{\s} G_1)v(t),
		$$
		by treating $v(t)$ as the independent input. The collected data satisfies \eqref{eq:sys-clean-trans} as desired. 
		This observation obviates the need to perform inversion for the matrix ($I+G_1$), which incurs the cost of $\calO(m^3)$, and is imposed on either the pre- or post-processing  mechanism. 
	\end{remark}
	
	The transformation \eqref{eq:pre-proc-func} constitutes the pre-processing mechanism, i.e. $f_0$, of the proposed privacy scheme. 
	Next, we ask Cloud to solve the data-driven stabilization problem, but using the transformed data $(X_1,X_0, V_0)$. This brings us to the following result.
	\begin{proposition}\label{prp:clean-data-open-loop}
		Let Cloud receive the data matrices $(X_1, X_0, V_0)$ given in \eqref{eq:data-original-ol} and \eqref{eq:pre-proc-func}. Furthermore, let Cloud solve the LMI \eqref{eq:final_LMI-for_QMI} with matrices $(\bA,\bB,\bC)$ given by
		\begin{equation}\label{eq:clean-data-mat}
		{\bA} = I_{n+m},\quad  
		{\bB}^\top = -{X}_1\bbm {X}_0 \\ V_0\ebm^{\dagger},
		\quad  {\bC} = {\bB}^\top {\bB}.
		\end{equation}
		Assuming that this LMI is feasible, denote it solution by  $(\bar{Y}, \bar{P})$ and set $\bar{K}:=\bar{Y}\bar{P}^{-1}$.
		Then the following holds:
		\begin{enumerate}[(i)]
			\item Stability:  The matrix $A_\s + B_\s K_\s$ is stable where $K_\s = F_1 + (I + G_1)\bar{K}$. \label{prp:stability-clean-data-phase-1}
			\item Open-loop privacy: The matrices $A_\s$ and $B_\s$ remain private against Cloud.\label{prp:privacy-open-clean-data-phase-1}
		\end{enumerate}
	\end{proposition}
	\noindent \textit{Proof of \ref{prp:stability-clean-data-phase-1}:} 
	By Lemma \ref{lem:stability-QMI}, the matrix $A+BK$ is stable for any $(A, B)$ belonging to the set 
	\begin{equation*}
	\begin{aligned}
	{\calC} :=\{\bbm {A} & {B}\ebm^\top :={Z}\mid {\bC}+{\bB}^{\top}{Z}+{Z}^{\top}{\bB}+ {Z}^{\top}{\bA} {Z} \preceq 0 \},
	\end{aligned}
	\end{equation*}
	with $\bA, \bB, \bC$ given by \eqref{eq:clean-data-mat}. 
	{It is easy to see that above set is a singleton and equals to ${\calC}=\{-(\bA^{-1}\bB)\}$.} By \eqref{eq:sys-clean-trans}, it follows that
	$$
	\underline{Z} := \bbm A_{\s} + B_\s F_1&B_{\s} + B_{\s} G_1 \ebm^\top \in {\calC}.
	$$
	Since the program \eqref{eq:final_LMI-for_QMI} is feasible with $\bar{K}=\bar{Y}\bar{P}^{-1}$, we conclude that $A_\s + B_\s(F_1 + (I + G_1)\bar{K})$ is stable. 
	\par \noindent\textit{Proof of \ref{prp:privacy-open-clean-data-phase-1}:}
	The set of matrices that Cloud receives from us is
	\be\label{eq:cloud_knowlege}
	\mathcal{I}_{\text{Cloud}}:=\{X_1, X_0, V_0\}.
	\ee
	By Assumption \ref{assum:rank-data}, the only  pair that is consistent with the received data is
	\be\label{eq:consistent_pair}
	\bar{A}=A_{\s} + B_\s F_1, \quad \bar{B}=B_{\s} + B_{\s}G_1,
	\ee
	where $(\bar{A}, \bar{B})$ is the unique solution to $X_1 = \bar{A}X_0 + \bar{B}V_0$. Note that $(\bar{A}, \bar{B})$ can be computed by Cloud from the data $\mathcal{I}_{\text{Cloud}}$.
	Denote $G:=I+G_1$ and consider the set
	\be\label{eq:quadruple_set}
	\calQ(\tilde F_1, \tilde{G}) := 
	\{
	(\tilde{G}^{-1}\tilde{F}_1, \tilde{G}^{-1} G) \},
	\ee
	where 
	$\tilde{F}_1 \in \R^{m \times n}$ and $\tilde{G} \in \R^{m \times m}$ are random matrices. 
	It can be verified that for any arbitrary pair $(\tilde{F}_1, \tilde G)$, the system matrices $(\hat A_\s , \hat B_\s)=(\bar{A}-B_\s \tilde{F}_1, B_\s\tilde{G})$ together with the matrices $(\hat{F}_1, \hat G)\in \mathcal{Q}(\tilde{F}_1, \tilde G)$ satisfy the equalities
	\[
	\bar{A}=\hat A_{\s} + \hat B_\s \hat F_1, \quad \bar{B}=\hat B_{\s} + B_{\s}G_1,
	\]
	and thus are consistent with \eqref{eq:consistent_pair}.
	%
	Therefore, the true $A_\s$ and  $B_\s$  remain unknown to Cloud. This completes the proof.  
	\EP
	
	In the following remark, we establish the relation of the transformation to secure computation literature.
	\begin{remark} 
		It can be observed from Proposition \ref{prp:clean-data-open-loop} that what Cloud identifies by receiving $(X_0,X_1, V_0)$
		depends on $\im B_\s$. In other words, $A_\s$ and $B_\s$ are unknown up to $\im B_\s$ for Cloud.  In fully actuated systems, we have $m =n$ and $\im B_\s=\mathbb{R}^n$, and thus Cloud can only identify the matrices $\tilde{A} \coloneqq A_{\s} + \DA_{\s}$ and $\tilde{B} \coloneqq B_{\s} + \DB_{\s}$
		where  $\DA_\s \in \R^{n\times n}$ and $\DB_\s \in \R^{n \times n}$ are arbitrary.
		In this case, the scheme becomes the one-time pad encryption scheme in cryptography \cite[Thm.2.9]{katz} where it is perfectly secret. We refer the reader to \cite[Def.2.3]{katz} for the exact definition of a perfectly secret encryption scheme, but generally speaking in this kind of schemes it is impossible to distinguish an encryption of a message $m$ (matrices $A_\s$ and $B_\s$) from an encryption of $\hat m$ (matrices $\hat{A}_\s$ and $\hat{B}_\s$).
	\end{remark}
	
	\vspace{0.15cm}
	\par Bearing in mind the definition of $K_\s$, we have $A_\s+B_\s K_\s=\bar A+\bar B \bar K$ with $(\bar{A}, \bar{B})$ being the unique solution to 
	\be\label{eq:Abar-Bbar-data}
	X_1 = \bar{A}X_0 + \bar{B}V_0.
	\ee
	As Cloud knows the matrices $(\bar A, \bar B)$ together with the controller $\bar K$, it retrieves the closed-loop matrix, 
	which is a breach of privacy with regard to P3 in Problem \ref{problem:data-driven}. Hence, Proposition \ref{prp:clean-data-open-loop}, by itself, does not provide a solution to Problem \ref{problem:data-driven} and it requires modification. 
	In the next subsection, we change the algorithm $f_1$ executed by Cloud and propose a suitable post-processing $f_2$ of the matrix $\bar{K}$ to patch this privacy hole.

	\subsection{Preserving privacy for the closed-loop system}\label{subsec:closed-loop-clean}
	Let $\bar Z:= \bbm \bar A & \bar B\ebm^\top$, where $\bar A$ and $\bar B$ are uniquely obtained from \eqref{eq:Abar-Bbar-data}.
	For a given $\gamma\geq0$, define
	\be\label{eq:clean-data-closed-loop}
	\begin{aligned}
		{C}(\gamma):= \{ \bbm A & B\ebm^\top =:Z& \mid\norm { Z- \bar Z}  \leq \gamma
		\}.
	\end{aligned}
	\ee  
	The set ${C}(\gamma)$ can be written in the form of $\Sigma$ given in \eqref{eq:ellips-lemma} as  
	$$
	{C}(\gamma)=\{\bbm {A} & {B}\ebm^\top :={Z}\mid {\bC}+{\bB}^{\top}{Z}+{Z}^{\top}{\bB}+ {Z}^{\top}{\bA} {Z} \preceq 0 \},
	$$
	with 
	\be \label{eq:clean-data-mat-closed}
	{\bA} = I_{n+m},\,\,  
	{\bB}^\top = -{X}_1\bbm {X}_0 \\ V_0\ebm^{\dagger},
	\,\,  {\bC} = {\bB}^\top {\bB} -\gamma^2 I_n.
	\ee
	Based on the set ${C}(\gamma)$, we define the following semi-definite program (SDP):
	\be \label{eq:dd-max-k-for-Ce}
	\begin{aligned}
		\max_{P, Y, \gamma\geq0} &\quad \gamma \\
		& \text{s.t.}\,\,  	   	\eqref{eq:final_LMI-for_QMI} \,\, \text{with} \,\, (\bA, \bB, \bC)\,\, \text{in } \,\, \eqref{eq:clean-data-mat-closed}.
	\end{aligned}
	\ee
	Note that the program \eqref{eq:dd-max-k-for-Ce} is feasible with $\gamma>0$ as long as  the system \eqref{eq:lin_sys} is stabilizable.
	In the following result, we show how solving  \eqref{eq:dd-max-k-for-Ce} by Cloud and suitably modifying its solution provides us with the desired results.
	
	\begin{proposition}\label{prp:closed-loop-privacy} 
		Let Cloud receive the data matrices $(X_1, X_0, V_0)$ given in \eqref{eq:data-original-ol} and \eqref{eq:pre-proc-func}. Furthermore,   let Cloud solve the SDP \eqref{eq:dd-max-k-for-Ce} with the solution $(\bar P, \bar Y, \bar \gamma)$ and $\bar \gamma >0$.
		Let $\bar{K}:=\bar{Y}\bar{P}^{-1}$, and select
		$F_2 \in \R^{m \times n}$ and $G_2 \in \R^{m \times m}$ 
		such that 
		\be \label{eq:imposing-deltaF-deltaG}
		\norm {\bbm F_2 - F_1 & G_2 - G_1 \ebm} \leq \frac{\bar{\gamma}}{\norm{B_{\s}}},
		\ee
		and 
		\be \label{eq:not-equal-closed-loop}
		\bbm F_2 - F_1 & G_2 - G_1\ebm \bbm I \\ \bar{K} \ebm \neq 0. 
		\ee
		Choose
		\be \label{eq:final-cont-clean-data}
		K_\s = F_2  + (I + G_2) \bar{K}.
		\ee
		Then the following statements hold:
		\begin{enumerate}[(i)]
			\item Stability:  The matrix $A_\s + B_\s K_\s$ is stable. \label{prp:stability-clean-data-phase-2}
			\item Open-loop privacy: The matrices $A_\s$ and $B_\s$ remain private against Cloud. \label{prp:privacy-open-clean-data-phase-2}
			\item Closed-loop privacy: The matrix $A_\s + B_\s K_\s$ remains private against Cloud. \label{prp:privacy-closed-clean-data-phase-2}
		\end{enumerate}
	\end{proposition}
	\noindent \textit{Proof of \ref{prp:stability-clean-data-phase-2}:}
	It follows from \eqref{eq:imposing-deltaF-deltaG} that
	$$
	\norm {\bbm B_{\s}(F_2 - F_1) & B_{\s}(G_2 - G_1) \ebm } \leq \bar{\gamma},
	$$
	and thus
	\begin{align*}
	||&\bbm A_{\s}+B_{\s}F_2 & B_{\s}+B_{\s}G_2 \ebm - \\
	&\qquad \bbm A_{\s}+B_{\s}F_1 & B_{\s}+B_{\s}G_1\ebm|| \leq \bar{\gamma}.
	\end{align*}
	Therefore, bearing in mind \eqref{eq:sys-clean-trans}, \eqref{eq:Abar-Bbar-data}, and \eqref{eq:clean-data-closed-loop}, the pair
	$
	\bbm A_{\s}+B_{\s}F_2 & B_{\s}+B_{\s}G_2\ebm \in {C}(\gamma)
	$ 
	and hence the controller $\bar{K}$ stabilizes it, i.e., the matrix
	$$
	A_{\s}+B_{\s}F_2 + (B_{\s}+B_{\s}G_2)\bar{K}
	$$
	is stable. By \eqref{eq:final-cont-clean-data}, we conclude that the matrix $A_\s + B_\s K_\s$ is stable. 
	\par \noindent \textit{Proof of \ref{prp:privacy-open-clean-data-phase-2}:} 
	By noticing that Cloud receives the same data matrices as in \eqref{eq:cloud_knowlege}   the proof reduces to the proof of 
	Proposition \ref{prp:clean-data-open-loop}.\ref{prp:privacy-open-clean-data-phase-1}
	\par \noindent\textit{Proof of \ref{prp:privacy-closed-clean-data-phase-2}:} To prove this statement, first recall that the matrices $F_1$ and $G_1$ are randomly selected in \eqref{eq:pre-proc-func}. 
	Define the matrices
	\bse
	\begin{align}
	\bar{A}_{cl} &:= A_{\s}+B_{\s}F_1+( B_{\s} + B_{\s}G_1)\bar{K} \label{eq:Abarcl} \\
	\Delta  &:= B_\s (F_2 - F_1) + B_\s(G_2 - G_1)\bar{K} \label{eq:Delta}.
	\end{align}
	\ese
	By 
	\eqref{eq:not-equal-closed-loop}, it follows that $\Delta   \neq 0 $. In addition, note that
	\be\label{eq:delta-m-role}
	\bar{A}_{cl}  + \Delta = A_\s + B_\s K_\s.
	\ee
	Note that $\bar{A}_{cl}$ is revealed to Cloud (see \eqref{eq:Abar-Bbar-data} and the subsequent discussion) 
	and that $\bar{K}$ is known to Cloud.  
	Nevertheless, since $\Delta\neq0$ and unknown to Cloud then the closed-loop system $A_\s + B_\s K_\s$ remains private against Cloud. 
	\EP
	\begin{remark}\label{rem:closwd-loop-budget}
		By working with the SDP \eqref{eq:dd-max-k-for-Ce}, rather than simply stabilizing the pair $(\bar A, \bar B)$ as in Proposition \ref{prp:clean-data-open-loop}, we create an additional uncertainty layer. The extent of this uncertainty is proportional to $\bar \gamma$ returned by Cloud. A higher value provides a larger set for the designer to pick the second pair of random matrices $F_2$ and $G_2$ in \eqref{eq:imposing-deltaF-deltaG}. The additional constraint in \eqref{eq:not-equal-closed-loop} is to ensure that the designer does not pick a pair that results in the same closed-loop system as before, i.e. $\bar A+ \bar B \bar K$.
	\end{remark}
	\begin{remark}
		While selecting $F_2$ and $G_2$ in \eqref{eq:imposing-deltaF-deltaG} requires the value of $\norm{B_{\s}}$, the results remain to hold if any known upper bound of $\norm{B_\s}$ is used in \eqref{eq:imposing-deltaF-deltaG}. Nevertheless, the accuracy of this upper bound will affect the size of uncertainty set mentioned in Remark \ref{rem:closwd-loop-budget}.
		%
	\end{remark}

	We summarize the results  
	to provide a solution to Problem \ref{problem:data-driven}.
	\begin{theorem} \label{thrm:main-clean-data}
		Let the data set \eqref{eq:data-original-ol} be collected from the system \eqref{eq:lin_sys}. 
		Then the pre-processing mechanism $f_0$ given in \eqref{eq:pre-proc-func}, the algorithm $f_1$ given in \eqref{eq:dd-max-k-for-Ce} and the post-processing mechanism $f_2$ given in \eqref{eq:imposing-deltaF-deltaG}-\eqref{eq:final-cont-clean-data} provide a solution to Problem \ref{problem:data-driven}.
	\end{theorem}
	\begin{proof}
		The proof follows from the proof of Propositions \ref{prp:clean-data-open-loop} and \ref{prp:closed-loop-privacy}. 
	\end{proof}

	
	
	\section{Preserving privacy for systems with disturbance}\label{sec:privacy-disturb-sys}
	
	In the previous section, we assumed the data was ``clean". However, in practice, data is often subject to disturbances. To address this, we now consider a scenario where the data is collected from the following system (see also \cite{bisoffi2022data} and \cite{van2020noisy}):
	\be\label{eq:lin-sys-noise}
	x(t+1) = A_{\s}x(t) + B_{\s}u(t) + d(t),
	\ee 
	where $d(t)$ is an unknown disturbance. 
	The presence of the disturbance per se in the system dynamics requires a redesign of the algorithm $f_1$ executed by Cloud to obtain a stabilizing controller. 
	In terms of privacy, while the disturbance $d(\cdot)$ in \eqref{eq:lin-sys-noise} can potentially create uncertainties for Cloud  regarding the open-loop and closed-loop matrices, we do not consider $d(\cdot)$ as an enhancing factor in the privacy-preserving scheme. 
	One of the reasons is that de-noising techniques have been adopted in the literature on data driven control to decrease the effects of the disturbance on the collected data (see, e.g., \cite[Sec.6]{de2021low}). Furthermore, as studied in \cite{bisoffi2021trade}, the identification of a system through data can result in a consistency set which can be substantially smaller than the set on which the control design is based (see the case study in \cite[Sec.5]{bisoffi2021trade}).
	Therefore, as a privacy-aware approach, we ensure that  privacy guarantees required by Problem \ref{problem:data-driven} are independent of both the presence of disturbance and Cloud's strategy in estimating privacy-sensitive parameters.

	\par Due to presence of disturbance, the equation \eqref{eq:open-loop-data-system} modifies to
	\be\label{eq:open-loop-noisy-system}
	X_1=A_\s X_0+ B_\s U_0+D_0,
	\ee
	where 
	$$ D_0  \coloneqq \bbm d(0) & d(1) & \cdots & d(T-1)\ebm \in \R^{n \times T}.$$
	We consider the disturbance class 
	\be \label{eq:dist-model}
	\calD \coloneqq \{D \in \R^{n \times T}\mid D D^\top \preceq \Delta \Delta^\top\},
	\ee 
	for some known $\Delta \in \R^{n \times s} $. 
	The matrix $\Delta$ is based on prior information that we have over $d(\cdot)$.
	While the matrix $D_0$ is unknown, we assume that it belongs to the set in \eqref{eq:dist-model} as formally stated next:
	\begin{assumption}\label{assum:noise-model}
		$D_0 \in \calD$.
	\end{assumption}
	The disturbance set \eqref{eq:dist-model} captures notable classes of $d(\cdot)$ in system and control as has been studied in the data-driven control literature; see e.g. \cite{bisoffi2022data}, \cite{van2020noisy}, and \cite{berberich2020robust}.
	By applying the transformation \eqref{eq:pre-proc-func} to the data in \eqref{eq:open-loop-noisy-system}, we have that
	{\be\label{eq:X1X0V0D0}
		X_1 = \underbrace{(A_\s + B_\s F_1)}_{=:A}X_0 + \underbrace{(B_\s + B_\s G_1)}_{=:B}V_0 + D_0.
		\ee}
	By isolating $D_0$ and substituting it in  \eqref{eq:dist-model},
	we can write the consistency set for $(A, B)$ as
	\be\label{eq:ellipsoid-first-form}
	\begin{aligned}
		{\calC} :=\{\bbm {A} & {B}\ebm^\top :={Z}\mid {\bC}+{\bB}^{\top}{Z}+{Z}^{\top}{\bB}+ {Z}^{\top}{\bA} {Z} \preceq 0 \},
	\end{aligned}
	\ee
	with the matrices obtained from transformed data  
	\be\label{eq:noisy-data}
	\begin{aligned}
		&{\bA} = \bbm {X}_0 \\ V_0\ebm\bbm {X}_0 \\ V_0\ebm^{\top},\quad  
		{\bB} = -\bbm {X}_0 \\ V_0\ebm {X}_1^{\top},\\
		& \qquad \qquad {\bC} = {X}_1{X}_1^{\top} - \Delta \Delta^{\top}.
	\end{aligned}
	\ee
	We note that the set ${\calC} $ differs from the actual data consistency set of the system corresponding to \eqref{eq:open-loop-noisy-system}, thanks to the presence of the random matrices $F_1$ and $G_1$.
	Subsequently, we ask Cloud to solve the LMI \eqref{eq:final_LMI-for_QMI} with the matrices $\bA$, $\bB$, $\bC$ given above.
	An analogous statement to Proposition \ref{prp:clean-data-open-loop} can be stated for stability of the closed-loop matrix $A_\s + B_\s K_\s$ and the privacy of $(A_\s, B_\s)$. Nevertheless, preserving privacy of the closed-loop remains to be an issue, following a similar discussion as in Subsection \ref{subsec:closed-loop-clean} and noting that Cloud could potentially identify the matrix $Z_\s := (A_\s + B_\s F_1, B_\s + B_\s G_1)$ by removing the effect of $D_0$ in \eqref{eq:X1X0V0D0}.
	Next, building on the idea in \eqref{eq:dd-max-k-for-Ce}, we propose a mechanism that preserves privacy of the closed-loop.  
	\subsection{Preserving privacy for the closed-loop system with disturbance}
	
	
	For a given $\gamma\geq0$, define
	\be\label{eq:noisy-data-closed-loop}
	\begin{aligned}
		{C}(\gamma):= \{ \bbm A & B\ebm^\top =:Z \mid \norm { Z- \bar Z } \leq \gamma, 
		\text{for some}\,\, \bar Z \in {\calC} \},
	\end{aligned}
	\ee 
	where ${\calC}$ is given by \eqref{eq:ellipsoid-first-form}.
	This set serves as an extension of \eqref{eq:clean-data-closed-loop} to the case where disturbance is present in the data, which, with a slight abuse of notation, is denoted again by ${C}(\gamma)$. Loosely speaking, the set ${C}(\gamma)$ enlarges the data consistency set \eqref{eq:ellipsoid-first-form} by $\gamma$ in every direction; see Figure \ref{fig:offset}.   
	Now, we can form a similar SDP as in \eqref{eq:dd-max-k-for-Ce} to stabilize the pair $[A \; B]^\top\in {C}(\gamma)$ . However, the main challenge in designing such an algorithm is that  ${C}(\gamma)$ cannot in general be written in the form of a QMI and thus we cannot readily apply Lemma \ref{lem:stability-QMI}. To tackle this challenge, we provide an 
	over-approximation for ${C}(\gamma)$ in the form of a QMI in Proposition \ref{prp:ellip-Rn-same-Ab}. As the derivation of this over-approximation  relies on rewriting the set \eqref{eq:ellipsoid-first-form}, we present its proof along a required lemma in the Appendix. 
	\begin{figure}[h]
		\begin{center}
			\includegraphics[width=0.4\textwidth]{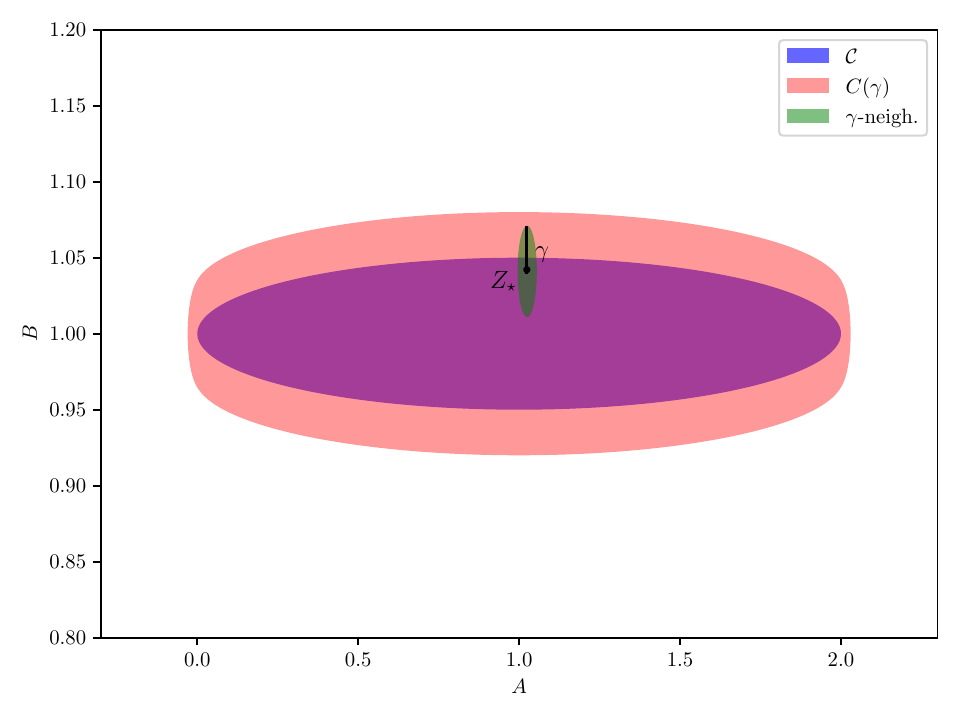}
			\caption{ The pair $Z_\s = (A_\s + B_\s F_1, B_\s + B_\s G_1)$ (bullet point), its $\gamma$-neighborhood systems (green), the consistency set ${\calC}$
				(blue) and the set ${C}(\gamma)$ (red) for a system with $n=m=1$. Note the red surface covers the blue surface.}\label{fig:offset}
		\end{center}
	\end{figure}
	\begin{proposition}\label{prp:ellip-Rn-same-Ab}
		Let ${C}(\gamma)$  be the set defined in \eqref{eq:noisy-data-closed-loop}.
		Define the set $\bar{\calC}$ as
		$$
		\bar{\calC} :=\{\bbm {A} & {B}\ebm^\top :={Z}\mid \bar{{\bC}}+{\bB}^{\top}{Z}+{Z}^{\top}{\bB}+ {Z}^{\top}{\bA} {Z} \preceq 0 \},
		$$
		with 
		\be\label{eq:Cbar}
		\bar{{\bC}} := \bC - (2 \gamma \norm{\bA^{\frac{1}{2}}}\norm{({\bB}^{\top}{\bA}^{-1}{\bB}-{\bC})^{\frac{1}{2}}} + \gamma^2 \norm{\bA})I_n.
		\ee
		Then, it holds that ${C}(\gamma) \subseteq \bar{\calC}$.
	\end{proposition}

	\par By having an over-approximation for the set \eqref{eq:noisy-data-closed-loop},  we define the optimization program
	\be \label{eq:dd-max-nosiy-data}
	\begin{aligned}
		\max_{P, Y, \gamma>0} &\quad \gamma \\
		& \text{s.t.}\,\,  	   	\eqref{eq:final_LMI-for_QMI} \,\, \text{with} \,\, (\bA, \bB, \bar{{\bC}})\,\, \text{in } \,\, \eqref{eq:noisy-data}, \eqref{eq:Cbar}.
	\end{aligned}
	\ee
	Next, we present the main result of this section:
	\begin{theorem} \label{thrm:main-noisy-data}
		Let the data $(X_0,X_1, U_0)$ be collected from the system \eqref{eq:lin-sys-noise}.
		Suppose that \eqref{eq:dd-max-nosiy-data} be feasible with the solution $(\bar P, \bar Y, \bar \gamma)$.
		Consider the pre-processing mechanism $f_0$ given in \eqref{eq:pre-proc-func}, the algorithm $f_1$ given in \eqref{eq:dd-max-nosiy-data} and the post-processing mechanism $f_2$ given in \eqref{eq:imposing-deltaF-deltaG}-\eqref{eq:final-cont-clean-data}. 
		
		Then:
		\begin{enumerate}[(i)]
			\item Stability:  The matrix $A_\s + B_\s K_\s$ is stable.
			\item Open-loop privacy: The matrices $A_\s$ and $B_\s$ remain private against Cloud.
			\item Closed-loop privacy: The matrix $A_\s + B_\s K_\s$ remains private against Cloud. 
		\end{enumerate}
	\end{theorem}
	\noindent \textit{Proof of (i):} Analogous to the proof of Proposition \ref{prp:closed-loop-privacy}.(i), the pair 	
	$\bbm A_{\s}+B_{\s}F_2 & B_{\s}+B_{\s}G_2\ebm \in {C}(\gamma)$ in \eqref{eq:noisy-data-closed-loop}.
	By Proposition \ref{prp:ellip-Rn-same-Ab}, we have that ${C}(\gamma) \subseteq \bar{\calC}$ and therefore the 
	aforementioned pair belongs to $\bar{\calC}$.
	By feasibility of the optimization \eqref{eq:dd-max-nosiy-data} and bearing in mind \eqref{eq:final-cont-clean-data}, it follows that  
	the matrix $A_\s + B_\s K_\s$ is stable.    
	\par \noindent \textit{Proof of (ii):} It is sufficient to prove the statement for the case where Cloud knows the disturbance matrix $D_0$. This assumption reduces the proof to the proof of Proposition \ref{prp:clean-data-open-loop}.(ii) with the only difference that $(\bar{A}, \bar{B})$ in \eqref{eq:consistent_pair} is now the unique solution to $X_1 -D_0 = \bar{A}X_0 + \bar{B}V_0$.
	\par \noindent \textit{Proof of (iii):} Again, it is sufficient to prove the statement for the case where Cloud knows $D_0$. 
	Under this assumption, Cloud infers the value of $\bar{A}_{cl}$ in \eqref{eq:Abarcl}, but not the nonzero matrix $\Delta$ in \eqref{eq:Delta}. This completes the proof (see also the proof of Proposition \ref{prp:closed-loop-privacy}.(iii).) \EP
	\section{Case study}\label{sec:case_study}
	In order to examine the proposed privacy-preserving scheme,
	we consider a batch reactor system which its discrete linearized model is given by \cite{de2019formulas}
	
	\footnotesize
	\be\label{eq:batch-reactor}
	\left[ 
	\begin{array}{c|c}
		A_\s & B_\s
	\end{array}
	\right]
	= 10^{-3}\left[ 
	\begin{array}{cccc|cc}
		1178 & 1 & 511 & -403 & 4 & -87\\
		-51 & 661 & -11 & 61 &  467 & 1 \\
		76 & 335 & 560 & 382 & 213 & -235\\
		0 & 335 & 89 & 849 & 213 & -16
	\end{array}
	\right].
	\ee
	\normalsize
	Note that the system is open-loop unstable and the true model is only used for collecting data.
	
	\par For generating the data set \eqref{eq:data-original-ol}\footnote{{Simulation files are available at \url{https://github.com/teimour-halizadeh/privacy_data_driven_control}}}, we apply to each input channel a random input sequence from the uniform distribution {$u \sim U(-5, 5)$} with length ${T=20}$ and we set {$x_0 \sim U(-2.5, 2.5)$}.
	For the privacy mechanism in Theorem \ref{thrm:main-clean-data} ($f_0$ mechanism), we randomly select entries of $F_1$ and $G_1$ from uniform distribution {$U(-1, 1)$} and then ask Cloud to solve the optimization \eqref{eq:dd-max-k-for-Ce}. The program is feasible and returns the solution
	\be\label{eq:Kbar-batch} 
	\bar{K} = 10^{-3}\bbm 1907 & -1439 & 1372 & -2721 \\
	3798 & 69 & 2305 & -2663
	\ebm, \,\,   \bar{\gamma} = 0.054
	\ee
	Then from Theorem \ref{thrm:main-clean-data} and assuming that $\norm{B_\s}$ is known, we select $F_2$ and $G_2$ such that \eqref{eq:imposing-deltaF-deltaG} and \eqref{eq:not-equal-closed-loop} hold and apply 
	$K_\s =  F_2 + (I +G_2)\bar{K}$ to the system. The controller $K_\s$ stabilizes system \eqref{eq:batch-reactor} as shown in 
	Figure \ref{fig:bias_injection} (blue line).
	\subsection{Preventing bias injection attack}
	As we argued in Subsection \ref{subsec:problem-setup}, an accurate model knowledge plays a crucial role in designing stronger attacks against a system. 
	We consider a scenario where either Cloud shares the information that it has obtained with an active adversary denoted as $\calA$ or Cloud has been compromised by $\calA$.   The set of parameters that $\calA$ receives from Cloud is
	$$
	\mathcal{I}_{\calA}:= \{X_1, X_0, V_0, \bar{\gamma}, \bar{K}\}. 
	$$
	Furthermore, based on Assumption \ref{assum:problem_structure}, $\calA$ knows the mechanisms $f_0$, $f_1$ and $f_2$. 
	\par The goal of $\calA$ is to inject a bias with high impact into steady-state trajectory of the system while remaining undetected. For reaching this goal, we assume that $\calA$ follows the procedure in \cite[Sec. 4.6]{teixeira2015secure}. In terms of disruption resources, $\calA$ has access to the actuator data of the system and  injects (adds) a signal $a(t)$ to $u(t)$ at time $t= T_{\text{inj}}$ and thus {allows} $u(t) + a(t)$ be applied to  system \eqref{eq:batch-reactor}. The form of the attack signal \cite[Eq. 19]{teixeira2015secure} for $t\geq T_{\text{inj}}$ is given by
	$$
	a(t+1) = {\beta}I_2 a(t) + (1 - \beta)I_2 a_{\infty},
	$$
	where {$a(T_{\text{inj}}) = 0$} and $0 < \beta < 1$. The main parameter that $\calA$ should design is $a_\infty \in \R^{m}$ for which it needs to consider two factors: First how to maximize its impact on the states of the system and second how to do so while remaining undetected. 
	\par The steady-state impact of $a(t)$ on the states of the closed-loop system is given by
	\be\label{eq:steady-state-attacker}
	x_\infty^{a}:=(I - \hat{A}_{cl})^{-1}\hat{B}_{cl} a_\infty,
	\ee
	where $\hat{A}_{cl} = A_\s + B_\s F_2 + (B_\s + B_\s G_2)\bar{K}$  and $\hat{B}_{cl} = B_\s$ are the true values of the closed-loop state matrix and the input matrix, respectively.
	\par On the system side we assume that the system is  equipped with an anomaly detection given by
	\be\label{eq:anomal-detect}
	r(t) = x(t), \quad t\geq T_a,
	\ee
	where $r(t)$ is the residual at time $t$ and $T_a$ is a sufficiently large time constant.
	We also consider the set
	\be\label{eq:stealthy-set}
	\mathcal{U} := \{ \norm{r(t)} \leq \delta_\alpha , \forall t \geq T_a\},
	\ee
	where $\delta_\alpha\geq 0$. Note that we assume  $\delta_\alpha$ is known to $\calA$. 
	
	\par For injecting a maximum-energy bias while remaining unnoticed (stealthy) the adversary $\calA$ thus needs to solve the optimization \cite[ Eq.21]{teixeira2015secure}
	\be \label{eq:max_bias_injection}
	\begin{aligned}
		\max_{a_\infty} &\quad \norm{x_\infty^{a}} \\
		& \text{s.t.}\quad  	   	r(t) \in \cal U.
	\end{aligned}
	\ee
	Note that the constraint at steady-state is equivalent to $\norm{x_\infty^{a}} \leq \delta_{\alpha}$.
	
	\par We set $\delta_\alpha = 0.2$ in set \eqref{eq:stealthy-set} and assume that $\calA$ starts injecting $a(t)$ with $\beta = 0.5$ at {$T_{\text{inj}} = 10$} by first attempting to solve the optimization \eqref{eq:max_bias_injection}. 
	
	In Table \ref{tab:who-knows-what}, we have specified  knowledge levels for $\calA$ and different bias injection policies that it adopts for solving \eqref{eq:max_bias_injection}.
	\begin{table}[h]
		\caption{knowledge of $\calA$ and its policies for bias injection}\label{tab:who-knows-what}
		\begin{tabular}{c|c}
			\hline
			True values &  \makecell{$B_{cl} = B_\s$ \\ $A_{cl} = A_\s + B_\s F_2 + (B_\s + B_\s G_2)\bar{K}$
			} \\ \hline
			Knowledge of $\calA$	& \makecell{$\hat{A} = A_\s + B_\s F_1$\\ $\hat{B} = B_\s + B_\s G_1$\\ $\bar{K}$, $\bar{\gamma}$, $\delta_\alpha$}  \\ \hline
			\multirow{4}{*}{Policy of $\calA$} & \makecell{(I): {exact knowledge} \\ $\hat{B}_{cl} = B_\s$ \\ $\hat{A}_{cl} = A_\s + B_\s F_2 + (B_\s + B_\s G_2)\bar{K}$}  \\ \cline{2-2} 
			&  \makecell{(II): {uses its estimation} \\ $\hat{B}_{cl} = B_\s + B_\s G_1$ \\ $\hat{A}_{cl} = A_\s + B_\s F_1 + (B_\s + B_\s G_1)\bar{K}$} \\ \cline{2-2} 
			&  \makecell{(III):  {knows $\norm{{B}_{\s}} $}   \\ {$\hat{B}_{cl} = ({\norm{B_\s}}/{\norm{B_\s + B_\s G_1}})(B_\s + B_\s G_1)$} \\ $\hat{A}_{cl} = A_\s + B_\s F_1 + (B_\s + B_\s G_1)\bar{K}$} \\ \cline{2-2} 
			&  \makecell{(IV): {knows ${B}_{\s} $}   \\ $\hat{B}_{cl} = B_\s$ \\ $\hat{A}_{cl} = A_\s + B_\s F_1 + (B_\s + B_\s G_1)\bar{K}$}
		\end{tabular}
	\end{table}
	
	\begin{figure}[h]
		\begin{center}
			\includegraphics[width=0.45\textwidth]{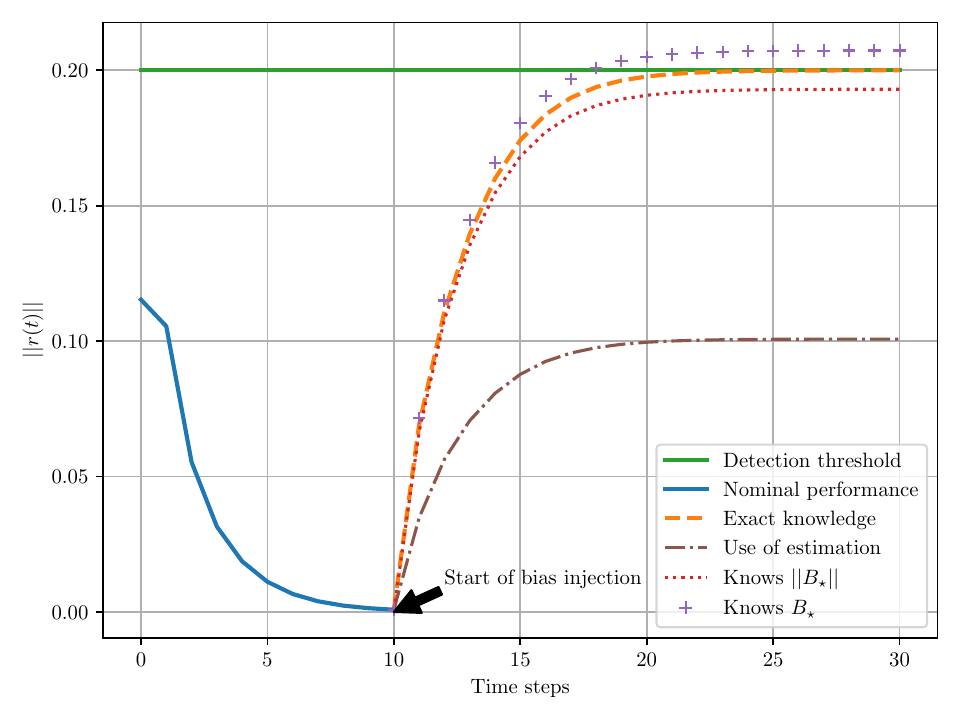}
			\caption{ Bias injection results by the active adversary $\calA$ using different system's model knowledge and policies.  }\label{fig:bias_injection}
		\end{center}
	\end{figure}
	{We have shown the results of this simulation at Figure \ref{fig:bias_injection}.}
	As it can be seen from this figure, 
	$\calA$ injects successfully the maximum impact bias (dashed line) when there is no privacy mechanism, i.e., when it knows $\hat{A}_{cl}$ and $\hat{B}_{cl}$ in \eqref{eq:steady-state-attacker}. 
	
	\par The dash-dotted line shows the response to $a_\infty$ when $\calA$ uses its estimation in Policy (II) for solving \eqref{eq:max_bias_injection}.
	It is clear that the steady-state value of $\norm{r(t)}$ remains substantially below threshold level $\delta_{\alpha}$ and thus does not have much impact on the system.
	In the next two scenarios, we suppose a situation when $\calA$ has obtained extra-information (side-knowledge) about the system \eqref{eq:batch-reactor}. First we consider the case when $\calA$ knows $\norm{B_{\s}}$. By using this information, $\calA$ modifies its estimation of $B_{cl}$ to $\hat{B}_{cl} = ({\norm{B_\s}}/{\norm{B_\s + B_\s G_1}})(B_\s + B_\s G_1) $; see Policy (II) in Table \ref{tab:who-knows-what}.
	The dotted line shows the response for this case. 
	As it can be seen, while the steady-state value of $\norm{r(t)}$ now has a higher value compared to previous case, it does not reach the maximum impact. 
	\par The next case that we consider is when 
	$\calA$ knows the exact value of $B_{\s}$. 
	As it can be seen form the figure (crossed line), the
	steady-state value of $\norm{r(t)}$  exceeds the threshold $\delta_{\alpha}$. The reason that even by knowing $B_\s$, Cloud still is unable to have the best possible attack is that the closed-loop system is unknown to it, which also emphasizes the importance of preserving privacy for the closed-loop system.
	These observations are in accordance with \cite[Sec.4.6]{teixeira2015secure} which argues that for implementing a successful bias injection the adversary needs to know the true steady-state dynamics of closed-loop system and detector.
	\subsection{The effects of disturbance $d(t)$ on $\bar{\gamma}$}
	As we observed in Proposition \ref{prp:closed-loop-privacy}, the value $\bar{\gamma}$ allowed the system designer to preserve privacy for the closed-loop system, and hence acted as a privacy budget. 
	This motivates us to study the relationship of $\bar{\gamma}$ and the disturbance magnitude. 
	To accomplish this, we consider the disturbance $d(\cdot)$ in \eqref{eq:lin-sys-noise} to have a uniform distribution $d(t) \sim U(-d_{\text{max}}, d_{\text{max}})$ where   {$d_{\text{max}} \in \{0, 0.02, \ldots, 0.16\}$} for each experiment. For each $d_{\text{max}}$, we collect $10^3$ data sets with the same parameters that we obtained \eqref{eq:Kbar-batch},  
	and compute $\bar{\gamma}$ by solving \eqref{eq:dd-max-nosiy-data}. We set $\Delta \Delta^\top = (n  d_{\text{max}}^2  T) I_n$ as it satisfies Assumption \ref{assum:noise-model}. 
	\begin{figure}[h]
		\begin{center}
			\includegraphics[width=0.49\textwidth]{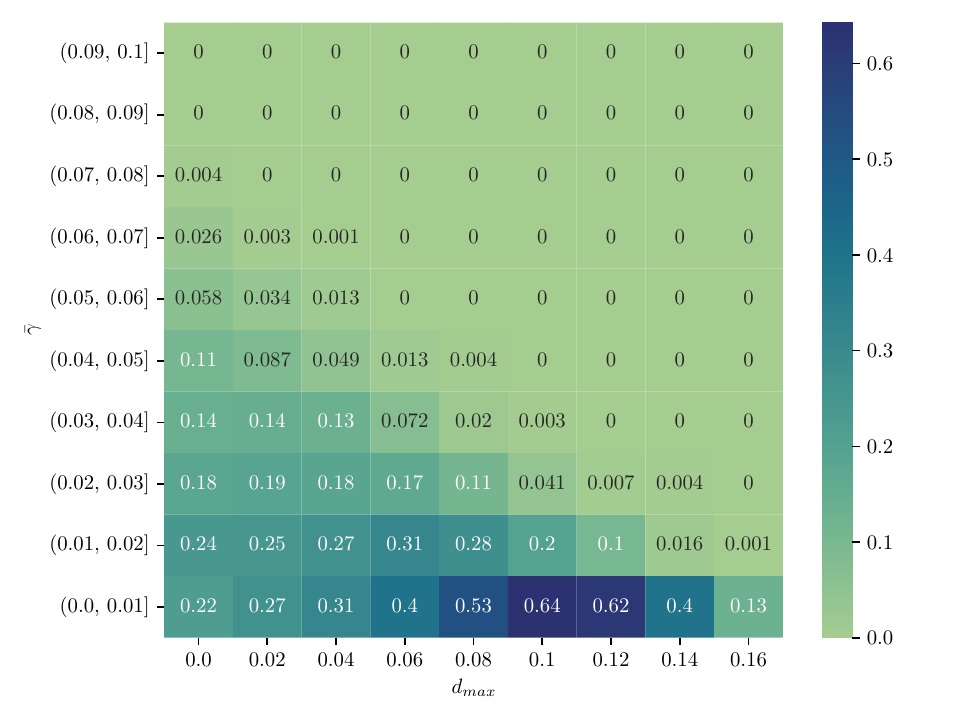}
			\caption{Results for the optimization \eqref{eq:dd-max-nosiy-data} with different values of $d_{\text{max}}$ (the upper-bound for the disturbance) and their corresponding ${\bar{\gamma}}$ (the privacy budget for the closed-loop system). The results should be read as, e.g.,
				we have obtained the ${\bar{\gamma}}$ as $0.02 < \bar{\gamma}\leq 0.03$  by solving \eqref{eq:dd-max-nosiy-data}  for $18\%$ of the collected data sets  when a uniform disturbance with $d_{\text{max}}=0.04$ was present in the system's dynamic.   
			}\label{fig:disturbance_epsilon_results}
		\end{center}
	\end{figure}
	
	\par The results are shown in Figure \ref{fig:disturbance_epsilon_results}. The value inside each square shows the fraction of the number of data sets to which
	$\bar\gamma$ returned from \eqref{eq:dd-max-nosiy-data} belongs. 
	We can see from this figure that $\bar{\gamma}$ decreases as $d_{\text{max}}$ increases.
	This means that the privacy budget is conversely related to the noise magnitude.  

	
	
	\section{Conclusion}\label{sec:conclusion}
	We have presented a scheme for outsourcing the computation of a data-driven stabilizing controller to a cloud computing service for an LTI system. 
	The scheme preserves privacy for the open-loop and closed-loop system matrices, can be adopted for both data without disturbance and with disturbance, is lightweight in computational over-head and does not degrade the performance of a system. We have demonstrated the effectiveness of the proposed scheme in a numerical case study of a batch reactor system.
	Future research directions for this study include extension of the results to signal tracking and optimal control problems and considering other models for Cloud.

	\section{Appendix: proof of Proposition \ref{prp:ellip-Rn-same-Ab}}
	To present the proof of Proposition \ref{prp:ellip-Rn-same-Ab}, we recall another form for the set ${\calC}$ given in \eqref{eq:ellipsoid-first-form}
	as (see \cite[Lem.2]{bisoffi2022data})
	\be \label{eq:ellipsoid-second-form}
	{\calC} =\{\bbm {A} & {B}\ebm^\top :={Z}\mid({Z}-{\bze})^{\top}{\bA}({Z}-{\bze})\preceq {\bQ}\},
	\ee
	where
	\be\label{eq:zeta-Q-data}
	{\bze} := -{\bA}^{-1}{\bB}, \quad    
	{\bQ}  := {\bB}^{\top}{\bA}^{-1}{\bB}-{\bC}.
	\ee
	\par In the next lemma, we remove the shift $\bze$ from the set ${\calC}$ and ${C}(\gamma)$.
	\begin{lemma}\label{lem:moving-zeta}
		Let ${\calC}$ and ${C}(\gamma)$ be the sets given in \eqref{eq:ellipsoid-second-form} and \eqref{eq:noisy-data-closed-loop}, respectively.
		Define the sets 
		\be\label{eq:ellips-matrix-Q-pd}
		\calC_o := \{\bbm {A} & {B}\ebm^\top :={Z}\mid Z^\top \bA Z\preceq \bQ \},
		\ee
		and 
		\be\label{eq:Cepsilon-ZtAZ}
		\begin{aligned}
			C_o(\gamma):= \{ \bbm {A} & {B}\ebm^\top :={Z} \mid \norm { Z- \underline{Z} } \leq \gamma, 
			\,\, \text{for some}\,\,\underline{Z} \in \calC_o  \}. 
		\end{aligned}
		\ee 
		Then the following statements hold:
		\begin{enumerate}[(i)]
			\item 			  $ \bar{Z}  \in \calC_o$ if and only if $\bar{Z} + \bze \in {\calC}$.
			\item 			 $\bar{Z}  \in  	C_o(\gamma)$ if and only if $\bar{Z} + \bze \in {C}(\gamma)$. 
		\end{enumerate}
	\end{lemma}
	\begin{proof}
		The proof for statement (i) is trivial. For statement (ii), 
		note that if $\bar{Z} \in {C}_o(\gamma)$ we have
		$$
		\norm { \bar{Z}- \underline{Z} } = \norm {\bar{Z}+\bze - (\underline{Z} +\bze)} \leq \gamma.
		$$
		Since $(\underline{Z} +\bze) \in {\calC}$ by statement (i),  then
		$\bar{Z}+\bze \in {C}(\gamma)$. The converse is analogous. 
	\end{proof}

	\par \noindent \textit{Proof of Proposition \ref{prp:ellip-Rn-same-Ab}:}
	We rewrite the set $\bar{\calC}$ as 
	\be\label{eq:over-matrix-Qb-zeta}
	\bar{\calC} = \{\bbm {A} & {B}\ebm^\top :={Z}\mid (Z - \bze)^\top \bA (Z- \bze) \preceq \bar{\bQ} \},
	\ee
	with the matrices
	\begin{align*}
	\bar{\bQ}   :=  \bQ    + 2 \gamma \norm{\bA^{\frac{1}{2}}}\norm{\bQ^{\frac{1}{2}}}I_n + \gamma^2 \norm{\bA}I_n,
	\end{align*}
	$\bA$ in \eqref{eq:noisy-data} and $\bze$ and $\bQ$ given in \eqref{eq:zeta-Q-data}. 
	We also remove the shift $\bze$ and define the set
	$$
	\bar{\calC}_{o} := \{\bbm {A} & {B}\ebm^\top :={Z}\mid Z^\top \bA Z \preceq \bar{\bQ} \}. 
	$$
	From Lemma \ref{lem:moving-zeta}, it follows that for proving ${C}(\gamma) \subseteq \bar{\calC}$ we equivalently can prove 
	$
	C_o(\gamma) \subseteq \bar{\calC}_{o}
	$. 
	\par 
	For any $\bar{Z} \in C_o(\gamma)$ there exists $\underline{Z} \in \calC_o$ such that
	$\norm{\bar{Z}- \underline{Z}} \leq \gamma$. 
Thus to prove ${C_o(\gamma) \subseteq 	\bar{\calC}_{o}}$, we consider the set
	\be\label{eq:offset-matrix-2}
	\begin{aligned}
		{\calS} = \{Z + \lambda \frac{\Delta}{\norm{\Delta}}\mid Z \in \calC_o,\,\, & 0 \leq \lambda \leq \gamma, \,\, \\
		& \text{and} \,\, \Delta \in \R^{(n+m) \times n} \},
	\end{aligned}
	\ee
	and prove ${{\calS} \subseteq 	\bar{\calC}_{o}}$ noting that $C_o(\gamma) \subseteq {\calS} $. Consider an arbitrary $\bar{Z} \in {\calS}$ as 
	\be\label{eq:Zbar-surface}
	\bar{Z} = Z + \lambda \frac{\Delta}{\norm{\Delta}},
	\ee
	with $Z \in \calC_o$ in \eqref{eq:ellips-matrix-Q-pd}. We rewrite \eqref{eq:Zbar-surface} as
	\be\label{eq:zs_boundary_CQ}
	\bar{Z} = \bA^{-1/2} \ups \bQ^{1/2} + \lambda \frac{\Delta}{\norm{\Delta}},
	\ee
	for some matrix $\ups$ with ${\ups}^{\top}\ups \preceq I$. By denoting $W_1:=\bQ^{1/2}\ups^\top\bA^{1/2}$, we have that 
	\be\label{eq:Zs-A-Zs-Q}
	\begin{aligned}
		&\bar{Z} ^\top \bA \bar{Z}  = \big(\bA^{-1}W_1^{\top} +  \lambda \frac{\Delta}{\norm{\Delta}}\big)^{\top} \bA \big(\bA^{-1}W_1^{\top} + \lambda \frac{\Delta}{\norm{\Delta}}\big) \\
		& =\bQ^{1/2}\ups^\top\ups \bQ^{1/2} + \lambda  \frac{W_1 \Delta}{\norm{\Delta}} +\lambda  \frac{(W_1 \Delta)^\top}{\norm{\Delta}} + \lambda^2 \frac{\Delta^\top \bA \Delta}{\norm{\Delta}^2}.
	\end{aligned}
	\ee
	Note that
	\begin{align*}
	& \norm{\lambda  \frac{W_1 \Delta}{\norm{\Delta}} +\lambda  \frac{(W_1 \Delta)^\top}{\norm{\Delta}} } \leq 2 \lambda \norm{W_1} \\
	& \norm{W_1} \leq \norm{\bQ}^{1/2} \norm{\bA}^{1/2},
	\end{align*}
	where the first inequality follows since  $\norm{\frac{\Delta}{\norm{\Delta}}} \leq 1$ and the second one since $\norm{\ups} \leq 1$. Hence,
	\be\label{eq:second-third-term-q-upper}
	\lambda  \frac{W_1 \Delta}{\norm{\Delta}} +\lambda  \frac{(W_1 \Delta)^\top}{\norm{\Delta}}  \preceq 2 \lambda \norm{\bQ}^{1/2} \norm{\bA}^{1/2} I_n.
	\ee
	Additionally, we have 
	\be\label{eq:fourth-term-q-upper}
	\lambda^2\frac{\Delta^\top \bA \Delta}{\norm{\Delta}^2} \preceq \lambda^2\norm{\bA} \frac{\Delta^\top \Delta}{\norm{\Delta}^2} \preceq \lambda^2\norm{\bA} I_n.
	\ee 
	Therefore, it follows from \eqref{eq:Zs-A-Zs-Q}-\eqref{eq:fourth-term-q-upper} and  $0 \leq \lambda \leq \gamma$ that
	$$
	\bar{Z}^\top \bA \bar{Z} \preceq \bQ    + 2 \gamma \norm{\bA^{1/2}}\norm{\bQ^{1/2}}I_n + \gamma^2 \norm{\bA}I_n = \bar{\bQ}.
	$$
	Thus $\calS \in \bar{\calC}_{o}$ and $C_o(\gamma) \in \bar{\calC}_{o}$. This completes the proof. \EP
	
	\bibliography{./MyReferences} 
\end{document}